\documentclass[a4paper, 11pt]{article}
\usepackage{a4wide,setspace,amsthm,amsfonts}
\usepackage{amsthm, verbatim, amsmath, tikz}
\usetikzlibrary{snakes}
\usepackage[ruled,vlined,linesnumbered,boxed,procnumbered,algosection,longend]{algorithm2e}
\usepackage{enumitem}

\newcommand{\m}[1]{$\mathcal{#1}$}

\newtheorem{theorem}{Theorem}[section]
\newtheorem{lemma}{Lemma}[section]
\newtheorem{observation}{Observation}[section]
\newtheorem{claim}{Claim}[section]
\newtheorem{remark}{Remark}[section]

\title{Maintaining Approximate Maximum Weighted Matching in Fully 
Dynamic Graphs}

\author{Abhash Anand\\
   Department of CSE,\\
   I.I.T. Kanpur, India\\
   {\small\texttt{abhash@cse.iitk.ac.in}} 
 \and
Surender Baswana\\
   Department of CSE,\\
   I.I.T. Kanpur, India\\
   {\small\texttt{sbaswana@cse.iitk.ac.in}}\footnote{
Research supported by the Indo-German Max Planck Center for
Computer Science (IMPECS).} 
 \and
Manoj Gupta\\
   Department of CSE,\\
   I.I.T. Delhi, India\\
{\small\texttt{gmanoj@cse.iitd.ernet.in}}
\and 
Sandeep Sen\\
   Department of CSE,\\
   I.I.T. Delhi, India\\
  {\small\texttt{ssen@cse.iitd.ernet.in}}
}

\begin{document}
\maketitle
\begin{abstract}
We present a fully dynamic algorithm for maintaining approximate 
maximum weight matching in general weighted graphs. 
The algorithm maintains a matching ${\cal M}$ whose weight is at least 
$\frac{1}{8} M^{*}$ where $M^{*}$ is the 
weight of the maximum weight matching. The algorithm achieves an expected 
amortized $O(\log n \log \mathcal C)$ time per edge insertion or deletion, 
where $\mathcal C$ is the ratio of the weights of the highest weight edge to 
the smallest weight edge in the given graph.
Using a simple randomized scaling technique,
we are able to obtain  a matching whith  expected approximation ratio 4.9108.
\end{abstract}
\section{Introduction}
Let $G=(V,E)$ be an undirected graph on $n=|V|$ vertices and $m=|E|$ edges. 
Let there be a weight function $w: E \rightarrow \mathbb{R}^+$ such that
$w(e)$, for any $e \in E$, represents the weight of $e$.
The weight function for a set of edges $M \subseteq E$ is represented by $w(M)$
and is defined as $\sum_{e \in M}w(e)$.
 
A subset $M$ of $E$ is a "matching" if no vertex of the graph is incident
on more than one edge in $M$. In an unweighted graph, a maximum 
matching is defined as the maximum cardinality matching (MCM).
In an weighted graph, maximum matching is defined as the maximum weight 
matching(MWM). For any $\alpha>1$, a matching is called 
$\alpha$-MWM($\alpha$-MCM) if it is at least $\frac{1}{\alpha}$ factor of 
MWM(MCM).

A dynamic graph algorithm maintains a data structure associated with some 
property (connectivity, transitive closure, matching) of a dynamic graph. 
The aim of a dynamic graph algorithm is to handle updates and answer 
queries associated with the set of vertices.
The updates in the graph can be insertion or deletion of edges ($V$ is assumed 
to be fixed).  
The dynamic algorithms which handle only insertions 
are called incremental algorithms and those that can handle only
deletions are called decremental algorithms. 
An algorithm that can handle both insertions and deletions
of edges is called a {\em fully dynamic} algorithm. 
In this paper, we present a fully dynamic algorithm for 
maintaining an approximate maximum weight matching.

\subsection*{Previous Results}

The fastest known algorithm for finding MCM in general graphs is by 
Micali and Vazirani\cite{micali1980v} that
runs in $O(m\sqrt n)$ time. Their algorithm can be used to compute
a matching having size $(1 - \epsilon)$ times the size of maximum matching in 
$O(m/\epsilon)$ time. 
Mucha and Sankowski\cite{mucha2004maximum} designed an algorithm that
computes MCM in $O(n^{\omega})$, where $\omega < 2.376$ is the exponent of $n$
in the fastest known matrix multiplication algorithm.
Relatively, fewer algorithms are known for maintaining
matching in a dynamic graph. The first algorithm was designed by 
Ivkovic and Lloyd\cite{ivkovi1994fully} with 
amortized update time $O(n + m)^{0.7072}$. 
Onak and Rubinfeld\cite{onak2010maintaining} presented an algorithm that
achieves expected amortized polylogarithmic update time and maintains an 
$\alpha$-approximate MCM where $\alpha$ was claimed to be some large 
constant but not explicitly calculated.
Baswana, Gupta and Sen\cite{baswana2011fully} presented a fully dynamic 
randomized algorithm for maintaining maximal matching in expected amortized 
$O(\log n)$ update time. It is well known that a maximal matching is a 2-MCM 
as well.

For computing maximum weight matching Gabow\cite{gabow1990data}
designed an $O(mn + n^2\log n)$ time algorithm.  
Preis\cite{preis1999linear} designed a $O(m)$ time algorithm for 
computing a $2$-MWM. Drake and Hougardy\cite{drake2003simple} designed
a simpler algorithm for the same problem. Vinkemeier and 
Hougardy\cite{vinkemeier2005linear} presented an algorithm to compute 
a matching which is $(2/3 - \epsilon)$ times the size of MWM in 
$O(m/\epsilon)$ time. Duan, Pettie and Su\cite{duan2011scaling} 
presented an algorithm to compute a matching which is 
$(1 - \epsilon)$ times the size of MWM in $O(m\epsilon^{-1}\log\epsilon^{-1})$ 
time. To the best of our knowledge, there have been no sub-linear algorithm for
maintaining MWM or approximate MWM in dynamic graphs.

\subsection*{Preliminaries}

Let $M$ be a matching in a graph $G=(V,E)$. A vertex in the graph 
is called \emph{free} with respect to $M$ if it is not incident on any 
edge in $M$. A vertex which is not free is called \emph{matched}. 
Similarly, an edge is called \emph{matched} if it is in $M$ and is called 
\emph{free} otherwise. If $(u, v)$ is a matched edge, then $u$ is called 
be the \emph{mate} of $v$ and vice versa. 
A matching $M$ is said to be \emph{maximal} if no edge can
be added to the matching without violating the degree bound of one for 
a matched vertex. An alternating path is
defined as a path in which edges are alternately matched and free, while an
augmenting path is an alternating path which begins and ends with free 
vertices. 
\subsection*{Our Results}

We present a fully dynamic algorithm 
that achieves expected amortized $O(\log n \log \mathcal{C})$ update time 
for maintaining $8$-MWM. Here $\mathcal C$ is the ratio of the weights of
the highest weight edge to the smallest weight edge in the given graph.
Our algorithm uses, as a subroutine, the algorithm of Baswana, Gupta and Sen 
\cite{baswana2011fully} for maintaining a maximal matching. 
We state the main result in \cite{baswana2011fully} formally
\begin{theorem}
Starting from an empty graph on $n$ vertices, a maximal matching in the 
graph can be maintained over any arbitrary sequence of $t$ insertion and deletion
of edges in $ O(t  \log n)$ time in expectation and 
$O(t\log n+ n\log^2 n)$ time with high probability.

\label{main-theorem}
\end{theorem}

Note that for the above algorithm,
the matching(random bits) at any time is not known to the 
adversary\footnote{The oblivious adversarial model is also used in randomized data-structure 
like universal hashing} for it to choose the updates adaptively.

The idea underlying
our algorithm has been inspired by the algorithm of
Lotker, Patt-Shamir, and Rosen\cite{lotker2007distributed} for maintaining
approximate MWM in distributed environment. Their algorithm 
maintains a 27-MWM in a distributed graph and achieves $O(1)$ rounds to 
update the matching upon any edge insertion or deletion.


\subsection*{Overview of our approach} 

Given that there exist very efficient algorithm \cite{baswana2011fully} 
for maintaining maximal matching (hence 2-MCM), it is natural to explore if
these algorithms can be employed for maintaining approximate MWM. 
Observe that MCM is a special case of MWM with all edges 
having the same weight. Since a maximal matching is 2-MCM, it can be observed 
that a maximal matching is 2-MWM in a graph if all its edges have the 
same weight. But this observation does not immediately extend to
the graphs having non-uniform weights on edges.  
Let us consider the case when the edge weights are within a range, say, 
$[\alpha^i, \alpha^{i + 1})$, where $\alpha >1$ is a constant. 
In such a graph the maximal matching gives a $2 \alpha$ approximation of 
the maximum weight matching. 
So, a maximal matching can be used as an approximation for MWM in a graphs 
where the ratio of weights of maximum weight edge to the smallest weight edge 
is bounded by some constant. To exploit this observation, we partition the
edges of the graph into levels according to their weight. We select a constant
$\alpha>1$ whose value will be fixed later on. Edges at level $i$ have weights 
in the range $[\alpha^i, \alpha^{i + 1})$ and the set of edges at level $i$ is 
represented by $E_i$, viz., 
$\forall e \in E_i, w(e) \in [\alpha^i, \alpha^{i + 1})$.

Observe that in this scheme of partitioning, any edge is present only at one 
level. The subgraph at level $i$ is defined as $G_i = (V, E_i)$. We maintain
a maximal matching $M_i$ for $G_i$ using the algorithm of Baswana, Gupta and 
Sen\cite{baswana2011fully}. 
The maximal matching at each level provides an approximation
for the maximum weight matching at that level. However,
$ \cup_i M_i$ is not necessary a matching since a vertex may have multiple 
edges incident on it from $ \cup_i M_i$. 
Let $\mathcal{H} = (V, \bigcup M_i)$ be the subgraph of $G$ 
having only those edges which are part of the maximal matching at some level. 
Our algorithm  maintains a matching in the subgraph $\mathcal{H}$ 
which is guaranteed to be $8$-MWM for the original graph $G$. 
The algorithm builds on the  algorithm in \cite{baswana2011fully}, though 
the analysis of algorithm for maintaining 8-MWM is not straightforward.

\section{Fully Dynamic $8$-MWM}
\label{fullydynamic}
Our algorithm maintains a partition of edges according to their levels. 
A maximal matching $M_i$ is maintained at each level using the fully dynamic 
algorithm in \cite{baswana2011fully}.  While processing any insertion or
deletion of an edge, this algorithm will leads to change in the status of edges
from being matched to free and vice-versa. This leads to deletion or insertion
of edges from/to ${\cal H}$. However, since the algorithm 
\cite{baswana2011fully} achieves expected amortized $O( \log n )$ time per
update, so the expected amortized number of deletions and insertions of edges
in ${\cal H}$ will also be $O( \log n )$ only. Our algorithm 
will maintain a matching \m M in the subgraph \m H taking advantage of the 
hierarchical structure of \m H. Since \m H is formed by the union of matchings
at various levels, a vertex can have at most one neighbor at each level. The
matching \m M is maintained such that for every edge of \m H  which is not in 
\m M there must be an edge adjacent to it at a higher level which is in \m M. 
For an edge $e$, let $Level(e)$ denote its level. In precise words, the 
algorithm maintains the following invariant after every update. 
\bigskip

$\forall e \in E(\mathcal H)$, either $e \in$ \m M or $e$ is adjacent to an edge
$e' \in$ \m M such that $Level(e') > Level(e)$.

\subsection*{Notations}

The algorithm maintains the following information at each stage.

\begin{itemize}
\item $M_l$ - A maximal matching at the level $l$.
\item $Free(v)$ - A variable which is true if $v$ is free in the matching \m M,
and false otherwise.
\item $Mate(v)$ - The mate of $v$, if it is not free.
\item $Level((u, v))$ or $Level(e)$ - The level at which the edge $e$
or the edge $(u, v)$ is present according to the condition that 
$\forall e \in G_i, w(e) \in [\alpha^i, \alpha^{i + 1})$.
\item $OccupiedLevels$ - The set of levels where there is at least one edge 
from \m H.
\item $L^{max}$ - The highest occupied level.
\item $L^{min}$ - The lowest occupied level.
\item $N(v,i)$ -  The neighbor of $v$ in $M_i$, if any, and $null$ otherwise. 
\item \m M - The matching maintained by our algorithm.
\end{itemize}

For a better understanding of our fully dynamic algorithm, the following 
section describes its static version for computing \m M in the graph \m H. 

\subsection*{Static Algorithm to obtain \m M from \m H}

\begin{procedure}
\caption{StaticCombine( )}
$\mathcal M = \phi$\;
\For{$i = L^{max}$ to $L^{min}$} {
	$\mathcal M = \mathcal M \cup M_i$\;
	\For{$(u, v) \in M_i$} {
		\For{$j = i - 1$ to $L^{min}$} {
			\For{$(x, y) \in M_j$} {
				\If {u = x or u = y or v = x or v = y} {
					$M_j = M_j \setminus \{(x, y)\}$\;
				}
			}
		}
	}
}
\end{procedure}

The static algorithm divides the edges of the graph $G$ into levels and 
 a maximal matching $M_i$ is obtained for each of the levels. Using these
maximal matchings we get the graph \m H. Thereafter the level numbers $L^{max}$ 
and $L^{min}$ are computed and the procedure \StaticCombine is used. 

The procedure \StaticCombine starts by picking all the edges in \m H at the 
highest level
and adds them to the matching \m M. For every edge $(u, v)$ added to the matching
\m M, all the edges in the graph \m H incident on $u$ and $v$ have to be
removed from the graph. 
The same process is repeated for the next lower level.
Note that every edge in \m H is either in the matching \m M or its
neighboring edge at some higher level is in \m M and thus 
the invariant is maintained. Observe that the matching
\m M is a maximal matching in \m H because of the way it is being computed.

\subsection*{Dynamic Algorithm to maintain \m M}

After each insertion or deletion of any edge, our algorithm maintains 
a matching \m M satisfying the invariant described above. 
Our algorithm processes insertions and deletions of edges in \m H to update 
\m M. An addition and deletion of the edges in \m H is caused due to 
addition/deletion of an edge in the original graph $G$. 
We describe some basic procedures first. Then the procedures for handling
addition and deletion of edges in \m H are described and finally the 
procedures for handling addition and deletion of edges in $G$ are described.

\begin{procedure}
\caption{AddToMatching($u, v$)}
$Free(u) = False$; $Free(v) = False$\;
$Mate(u) = v$; $Mate(v) = u$\;
$\mathcal M = \mathcal M \bigcup \{(u, v)\}$\;
\end{procedure}

\begin{procedure}
\caption{DelFromMatching($u, v$)}
$Free(u) = True$; $Free(v) = True$\;
$\mathcal M = \mathcal M \setminus \{(u,v)\}$\;
\end{procedure}

The procedure \AddToMatching adds an edge to the matching \m M updating the
free and mate fields accordingly. The procedure \DelFromMatching 
deletes an edge from the matching \m M updating the
mate and the free fields correctly. Both of them execute in $O(1)$ time.

\begin{procedure}[h]
\caption{HandleFree($u, lev$)}
\For{l from lev to $L^{min}$}{
	$v = N(u,l)$\;
	\If{$v$ is not null} {
		\If{v is free} {
			\AddToMatching($u, v$)\;
			\Return\;
		} \ElseIf{Level((v, Mate(v))) $<$ l} {
			$v' = Mate(v)$\;
			\DelFromMatching($v, v'$)\;
			\AddToMatching($u, v$)\;
			\HandleFree($v'$, $Level((v, v'))$)\;
			\Return\;
		}
	}
}
\end{procedure}

The procedure \HandleFree takes as an input a vertex $u$ which has become 
free in \m M and a level number $lev$ from where it has to start looking for a 
mate. Note that it follows from the invariant that $u$ does not have any free
neighbor at any level above $lev$. The procedure \HandleFree proceeds as
follows. It searches for a neighbor of $u$ in the decreasing order 
of levels starting from $lev$. In this process, on reaching a level $l \le lev$
if it finds a free neighbor of $u$, the corresponding edge is added to the 
matching \m M and the procedure stops. 
Otherwise if some neighbor $v$ is found which already
has a mate at some lower level than $l$, then notice that we are violating the 
invariant as $(u, v)$ does not belong to \m M and is neighboring to an edge 
in \m M at a lower level.
So, the edge $(v, Mate(v))$ is removed from the matching \m M, and the edge 
$(u, v)$ is added to the matching \m M. This change results in a free vertex 
which is at a lower level and so we proceed recursively to process it. 
Note that the recursive calls to \HandleFree are all with lower level 
numbers. So, the procedure takes $O(L^{max} -L^{min})$ time.

\begin{procedure}[h]
\caption{AddEdge($u, v$)}
$l = Level((u, v))$\;
$N(u,l) = v$\;
$N(v,l) = u$\;
\If{u is free and v is free} {
	\AddToMatching($u, v$)\;
} \ElseIf{u is free and v is not free} {
	\If {Level((u, Mate(u))) $<$ l} {
		$v' = Mate(v)$\;
		\DelFromMatching($v, v'$)\;
		\AddToMatching($u, v$)\;
		\HandleFree($v', Level((v, v'))$)\;
	}
} \ElseIf{u is not free and v is free} {		
	\If {Level((v, Mate(v))) $<$ l} {
		$u' = Mate(u)$\;
		\DelFromMatching($u, u'$)\;
		\AddToMatching($u, v$)\;
		\HandleFree($u', Level((u, u'))$)\;
	}
} \ElseIf {Level((v, Mate(v))) $<$ l and Level((u, Mate(u))) $<$ l} {
		$u' = Mate(u)$; $v' = Mate(v)$\;
		\DelFromMatching($u, u'$); \DelFromMatching($v, v'$)\;
		\AddToMatching($u, v$)\;
		\HandleFree($u', Level((u, u'))$)\;
		\HandleFree($v', Level((v, v'))$)\; 
} 
\end{procedure}

\begin{procedure}[h]
\caption{DeleteEdge($u, v$)}
$l = Level((u, v))$\;
$N(u,l) = null$\;
$N(v,l) = null$\;
\If{(u, v) $\in \mathcal M$} {
	\DelFromMatching($u, v$)\;
	\HandleFree($u, l$)\;
	\HandleFree($v, l$)\;
}
\end{procedure}

The procedure \AddEdge handles addition of edges to \m H. Suppose the edge 
$(u, v)$ is added to \m H. If both $u$ and $v$ are free with respect to \m M,
then the edge $(u, v)$ is added to the matching \m M. Otherwise, there must be
some edge(s) in \m M adjacent to $(u,v)$. This follows due to the fact that
\m M is a maximal matching in \m H.
If $(u, v)$ is adjacent to a higher level edge in \m M, then nothing is done.
If $(u, v)$ is adjacent to some lower level edge(s) in \m M, then notice that
the invariant maintained by the algorithm gets violated.
Therefore, we remove these lower level edge(s) (adjacent to $(u,v)$) 
from the matching \m M and adds the edge $(u, v)$ to the matching. 
At most $2$ vertices can become free due to the addition of this edge to \m M
and we handle them using the procedure \HandleFree.
If $u'$ was the previous mate of $u$, then the edge $(u,u')$ is removed from 
\m M. Since \m M satisfied the invariant before addition of this edge, 
all the neighboring edges of $u'$ at higher level than $Level(u, u')$ are 
matched to a vertex at higher levels. So $u'$ has to start looking for mates 
from the level of $(u, u')$. The procedure makes a constant number of calls 
to \HandleFree and thus runs in $O(L^{max} - L^{min})$ time.

The procedure \DeleteEdge does nothing if an unmatched edge from \m H is 
deleted. If a matched edge $(u,v)$ is deleted at level $l$, it calls 
\HandleFree for both the end points to restore the invariant in the matching. 
\HandleFree is called with the level $l$ because our invariant 
implies that all the neighbors of $u$ and $v$ are matched at higher levels.
So they cannot find a mate at higher levels.
This again takes $O(L^{max} - L^{min})$ time.

\begin{procedure}[h]
\caption{EdgeUpdate($u, v, type$)}
$l = Level((u, v)) = \lfloor \log_{\alpha} w(u, v) \rfloor $\;
\If {type is addition and $M_l$ is $\phi$} {
	$OccupiedLevels = OccupiedLevels \bigcup \{l\}$\;
	Update $L^{max}$ and $L^{min}$\;
}
Update $M_l$ using the algorithm in \cite{baswana2011fully}\;
\If {type is deletion and $M_l$ is $\phi$} {
	$OccupiedLevels = OccupiedLevels \setminus \{l\}$\;
	Update $L^{max}$ and $L^{min}$\;
}
Let \m D be the set of edges deleted from $M_l$ in step 5\;
Let \m A be the set of edges added to $M_l$ in step 5\;
\For{(x, y) $\in$ \m D} {
	\DeleteEdge($x, y$)\;
}
\For{(x, y) $\in$ \m A} {
	\AddEdge($x, y$)\;
}
\end{procedure}

The function \EdgeUpdate handles addition and deletion of an edge in 
$G$. It finds out the level of the edge and updates the maximal matching at that
level using the algorithm of Baswana, Gupta and Sen \cite{baswana2011fully}. 
It updates the {\em OccupiedLevels} set accordingly. This set is required
because the values of $L^{max}$ and $L^{min}$ are to be maintained. The algorithm
\cite{baswana2011fully} can be easily augmented to return the set of 
edges being added or deleted from the maximal matching in each update. 
As discussed before, expected amortized $O(\log n)$ edges 
change their status per update. Our algorithm processes these updates in 
\m H as described above. So, overall our algorithm has an
expected amortized update time of $O(\log n \cdot (L^{max} - L^{min}))$. 
Let $e^{max}$ and $e^{min}$ represent the edges having the maximum and the 
minimum weight in the graph. Recall that \m C $= w(e^{max})/w(e^{min})$.

$$L^{max} - L^{min} < \log_{\alpha} w(e^{max}) - \log_{\alpha} w(e^{min}) + 1 = O \left (\log\frac{w(e^{max})}{w(e^{min})} \right ) = O(\log \mathcal C)$$

So we can claim that
\begin{claim}
The expected amortized update time of the algorithm per edge insertion
or deletion is $O(\log n\log \mathcal C)$.
\end{claim}

In the next section we analyze the algorithm to prove that the matching 
\m M maintained by it at each stage is indeed $8$-MWM.

\subsection{Analysis}

To get a good approximation ratio, we bound the weight of $M^*$ with the weight of \m M. 
We now state a few
simple observations which help in understanding the analysis.

\begin{observation}
\label{obs1}
Since $M^*$ is a matching, no two edges of $M^*$ can be incident on the same vertex.
\end{observation}
\begin{observation}
\label{obs2}
For any edge $e \notin M^*$, there can be at most two edges of $M^*$ which are adjacent to $e$, one for each endpoint of $e$.
\end{observation}

To bound the weight of $M^*$ using the weight of \m M, we define a  many 
to one mapping $\phi : M^* \rightarrow {\cal M}$.
This mapping maps 
every edge in $M^*$ to an edge in \m M. Using this mapping, we find out all the
edges which are mapped to an edge $ e \in \mathcal M$ and bound their 
weight using the weight of $e$. Let this set be denoted by $\phi^{-1}(e)$.
For an edge $e^* \in M^*$, the mapping is defined as:

\begin{enumerate}
\item If $e^* \in E($\m H$)$ and $e^* \in$ \m M then  $\phi(e^*) = e^*$.
\item If $e^* \in E($\m H$)$ and $e^* \notin$ \m M then our invariant ensures 
that $e^*$ is adjacent to an edge
$e\in {\cal M}$ such that Level$(e)$ $>$ Level$(e^*)$. In this case, we define
$\phi(e^*)=e$. 
If $e^*$ is adjacent to two matched edges in \m M, map $e^*$ to any one of 
them. As a rule, if two edges are available for mapping, then we will map 
$e^*$ to any one of them.
\item If $e^* \notin E($\m H$)$, then consider its level, say $i$.
Since we maintain a maximal matching $M_i$ at level $i$, 
at least one of the end point of $e^*$ must be present in $M_i$. 
Let $e \in M_i$ be adjacent to $e^*$. If $e \in$ \m M, we define $\phi(e^*)=e$.
\item If $e^* \notin E($\m H$)$ and the edge $e\in M_i$ adjacent to $e^*$ 
is not present in ${\cal M}$ then $e$ must be adjacent to an edge 
$e' \in$ \m M such that Level$(e') >$ Level$(e)$. In this case, we define
$\phi(e^*)=e'$.
\end{enumerate}

Now that we have defined a many to one mapping, we find out the edges of $M^*$ which are
mapped to an edge $e \in$ \m M. An edge which is mapped to $e$ can either be $e$ itself or
be adjacent to $e$ or not adjacent to $e$.
If an edge of $M^*$, which is mapped to $e \in$ \m M,  is $e$ itself or is adjacent to $e$ , 
then it is called a \emph{Directly
mapped edge}. An edge of $M^*$ which is mapped to $e \in$ \m M and is not adjacent to $e$ is called an \emph{Indirectly
mapped edge}. Let $\phi^{-1}_{D}(e)$ and $\phi^{-1}_{I}(e)$ be the set of directly mapped and indirectly mapped edges
respectively for an edge $e \in M$. Directly mapped edges are of type $1$, $2$ and $3$ 
and indirectly mapped edges are of type $4$. An edge $e^* \in M^*$ can either be in $E($\m H$)$ or not. 
If it is in $E($\m H$)$, then it is mapped using type 1 and type 2 mapping else it is mapped using type 3 and
type 4 mapping. This implies all
the edges in $M^*$ are mapped by $\phi$.

If an edge $e \in$ \m M has an edge of type $1$ directly mapped to it, then 
$e$ will not have any other edge directly mapped to it. 
This follows from the definition of a directly mapped edge and 
Observation \ref{obs1}. There can be at most two directly mapped edges of the 
second type(Observation \ref{obs2}).
These edges mapped to $e$ are always from a level $< Level(e)$.
There can be at most two directly mapped edges of type $3$ also if they are not in \m H
but are adjacent to $e$. By Observation \ref{obs2}, there can only be two
such edges. 
\begin{claim} 
There can be at most two directly mapped edges to an edge $e \in \mathcal M$ at any level.
\label{claim1}
\end{claim}

The total weight of the
edges directly mapped to $e$ will be maximum
when both of them are from the same level as $e$.
Assume that $e$ is at level $i$. Summing the weights of the edges which are 
directly mapped to $e$, we get
\begin{equation}
\sum_{e^* \in \phi^{-1}_{D}(e)} w(e^*) < 2 * \alpha^{i + 1} < 2 \alpha w(e)
\label{eq1}
\end{equation}

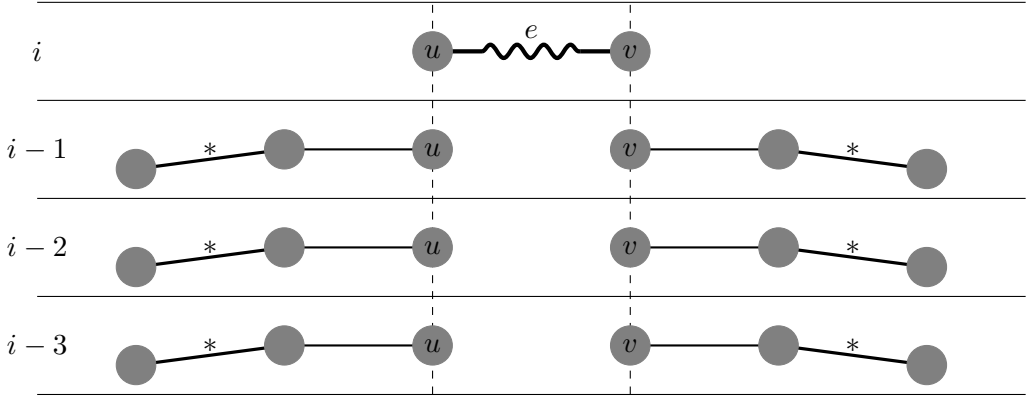
\begin{figure}
\begin{tikzpicture}[scale=1.3]
\tikzstyle matched=[ultra thick]
\tikzstyle vertical=[dashed]
\tikzstyle inH=[thick]
\tikzstyle inM=[very thick]
\draw (0, 0) -- (10, 0);
\draw (0, 1) -- (10, 1);
\draw (0, 2) -- (10, 2);
\draw (0, 3) -- (10, 3);
\draw (0, 4) -- (10, 4);
\draw[vertical] (4, 0) -- (4, 4);
\draw[vertical] (6, 0) -- (6, 4);
\draw[matched] (4, 3.5) -- (4.5, 3.5);
\draw[snake=snake,matched] (4.5, 3.5) -- (5.5, 3.5);
\draw[matched] (5.5, 3.5) -- (6, 3.5);

\draw[inH] (4, 2.5) -- (2.5, 2.5);
\draw[inH] (4, 1.5) -- (2.5, 1.5);
\draw[inH] (4, 0.5) -- (2.5, 0.5);
\draw[inH] (6, 2.5) -- (7.5, 2.5);
\draw[inH] (6, 1.5) -- (7.5, 1.5);
\draw[inH] (6, 0.5) -- (7.5, 0.5);

\draw[inM] (2.5, 2.5) -- (1, 2.3);
\draw[inM] (2.5, 1.5) -- (1, 1.3);
\draw[inM] (2.5, 0.5) -- (1, 0.3);
\draw[inM] (7.5, 2.5) -- (9, 2.3);
\draw[inM] (7.5, 1.5) -- (9, 1.3);
\draw[inM] (7.5, 0.5) -- (9, 0.3);

\filldraw[gray] (2.5, 2.5) circle (0.2cm);
\filldraw[gray] (2.5, 1.5) circle (0.2cm);
\filldraw[gray] (2.5, 0.5) circle (0.2cm);
\filldraw[gray] (7.5, 2.5) circle (0.2cm);
\filldraw[gray] (7.5, 1.5) circle (0.2cm);
\filldraw[gray] (7.5, 0.5) circle (0.2cm);

\filldraw[gray] (1, 2.3) circle (0.2cm);
\filldraw[gray] (1, 1.3) circle (0.2cm);
\filldraw[gray] (1, 0.3) circle (0.2cm);
\filldraw[gray] (9, 2.3) circle (0.2cm);
\filldraw[gray] (9, 1.3) circle (0.2cm);
\filldraw[gray] (9, 0.3) circle (0.2cm);

\filldraw[gray] (4, 3.5) circle (0.2cm);
\draw (4, 3.5) node {\textbf{$u$}};
\filldraw[gray] (6, 3.5) circle (0.2cm);
\draw (6, 3.5) node {$v$};

\filldraw[gray] (4, 2.5) circle (0.2cm);
\draw (4, 2.5) node {\textbf{$u$}};
\filldraw[gray] (6, 2.5) circle (0.2cm);
\draw (6, 2.5) node {$v$};

\filldraw[gray] (4, 1.5) circle (0.2cm);
\draw (4, 1.5) node {\textbf{$u$}};
\filldraw[gray] (6, 1.5) circle (0.2cm);
\draw (6, 1.5) node {$v$};

\filldraw[gray] (4, 0.5) circle (0.2cm);
\draw (4, 0.5) node {\textbf{$u$}};
\filldraw[gray] (6, 0.5) circle (0.2cm);
\draw (6, 0.5) node {$v$};

\draw (5, 3.7) node {$e$};

\draw (0, 3.5) node {$i$};
\draw (0, 2.5) node {$i - 1$};
\draw (0, 1.5) node {$i - 2$};
\draw (0, 0.5) node {$i - 3$};

\draw (1.75, 2.5) node {$*$};
\draw (1.75, 1.5) node {$*$};
\draw (1.75, 0.5) node {$*$};
\draw (8.25, 2.5) node {$*$};
\draw (8.25, 1.5) node {$*$};
\draw (8.25, 0.5) node {$*$};

\end{tikzpicture}
\caption{$e \in$ \m M. The edges marked $*$ are not in \m H and are in $M^*$. The edges
which are not marked $*$ are all in \m H.
All the edges marked by $*$ are indirectly
mapped to $e$. }
\label{figure1}
\end{figure}

Indirectly mapped edges can only be of the fourth kind in which the edge is not
in \m H, but is adjacent to an edge in \m H, which in turn is adjacent to $e$.
By definition, these edges are from a level lower than that of $e$. There can be at most
two edges from each level lower than Level$(e)$ which are in \m H and are adjacent
to $e$(see Figure \ref{figure1}). 
\begin{claim}
There can be at most two indirectly mapped edges to an edge $e \in \mathcal M$ at level $< Level(e)$.
\label{claim2}
\end{claim}
Note that there can be a large number of edges which are
indirectly mapped to $e$. Still we will be able to get a good bound on their total
weight. This is because there can be at most two indirectly mapped edges from
each level and the weight of edges in the levels decreases geometrically as we
go to lower levels.

Assume that $e$ is at level $i$. Summing the weight of edges which are indirectly mapped to $e$, we get
\begin{equation}
\sum_{e^* \in \phi^{-1}_{I}(e)} w(e^*) < 2 \sum_{j = i - 1}^{L^{min}} \alpha^{j + 1} < \frac{2 \alpha^{i + 1}}{\alpha - 1} < \frac{2\alpha w(e)}{\alpha - 1} 
\label{eq2}
\end{equation}
Thus, the total weight mapped to $e$ is -
$$\sum_{e^* \in \phi^{-1}(e)} w(e^*) = \sum_{e^* \in \phi^{-1}_{D}(e)} w(e^*) + \sum_{e^* \in \phi^{-1}_{I}(e)} w(e^*) < w(e) \left ( \frac{2 \alpha}{\alpha - 1} + 2 \alpha \right)$$

As reasoned before, an edge in $M^*$ is mapped to some edge in \m M. So summing this over all the edges in \m M, we get
$$\sum_{e \in  M} w(e) \left ( \frac{2 \alpha}{\alpha - 1} + 2 \alpha \right) > \sum_{e \in M}\sum_{e^* \in  \phi^{-1}(e)} w(e^*) = \sum_{e^* \in M^*}w(e^*)$$

The function $f(\alpha) = \left ( \frac{2 \alpha}{\alpha - 1} + 2 \alpha \right)$ attains its minimum value of $8$ at $\alpha = 2$. So,
if the value of $\alpha$ is picked to be $2$, we get an $8$ approximation
maximum weight matching algorithm. We can state the following theorem.

\begin{theorem}
 There exists a fully dynamic algorithm that maintains 8-MWM for any graph on $n$
  vertices in expected amortized $O(\log n \log \mathcal C)$ time per update.
\end{theorem}

\section{Improvements: Fully Dynamic 4.9108-MWM}
We use 
use the method of geometric rounding( see \cite{leah1}, \cite{leah2}) to 
reduce the approximation ratio to 4.9108.  

We choose a random number $r$ from $(0,1]$ and then partition 
the edges as follows. If an edge $e$ has weight 
$w(e) \in [ \alpha^{i+r}, \alpha^{i+r+1})$, then it belongs to level $i$ where $i \ge 0$.
Note that here we assume that the weight of an edge is always greater than $\alpha^{r}$.
This is generally not true but can be handled by initially multiplying the weight of
all edges by $\alpha$. From now on we assume that an edge will always belong to some level $i$.
Define $w_r(e) = \alpha^{i+r}$ if $e$ belongs to level $i$. The algorithm works on the
new weights instead of the original weight of the edge. The working of the algorithm 
is exactly same as in the previous section.

\begin{lemma}
For an edge $e$, $E_r[w_r(e)/ w(e)] = \frac{\alpha-1}{\alpha \log \alpha}$. Also $w(e) \ge w_r(e)$.
\label{randomexplemma}
\end{lemma}
\begin{proof}
Let $w(e) = \alpha^{i+\delta}$ where $i$ is an integer and $0 < \delta \le 1$. 

So $   w_r(e) = \begin{cases}
		\alpha^{r+i}, & \text{if} \ \ r \le \delta ,\\
        \alpha^{r+i-1}, & \text{if} \ \ r > \delta 
        \end{cases}
 $ 
 
The expected value can be calculated as:\\
\begin{tabular}{lllll}
$E_r[w_r(e)/w(e)]$ & = & $\displaystyle\int_{0}^{\delta} \frac{\alpha^{r+i}}{\alpha^{i+\delta}} dr + \int^{1}_{\delta} \frac{\alpha^{r+i-1}}{\alpha^{i+\delta}} dr$\\
 			& = & $\displaystyle\frac{1}{\alpha^{\delta} \ln \alpha} \Big(\big(\alpha^{r}\big)_{0}^{\delta} + \big( \alpha^{r-1} \big)_{\delta}^{1}\Big)$\\
 			& = & $\displaystyle\frac{1}{\alpha^{\delta} \ln \alpha} \Big(\big(\alpha^{\delta}-1\big) + \big( 1-\alpha^{\delta-1} \big)\Big) $\\\\
 			& = & $\displaystyle\frac{\alpha-1}{\alpha \ln \alpha}$
\end{tabular}

For the second statement of the lemma, if $e$ is at level $i$
then $w(e) \in [ \alpha^{i+r}, \alpha^{i+r+1})$ and $w_r(e) = \alpha^{i+r}$.
So $w(e) \ge w_r(e)$.
\end{proof}

Following the algorithm in the previous section but working on the new weights $w_r(e)$, we 
observe that all the edges at level $i$ have same new weights. The problem we faced 
in the previous section was that if $e \in M$ was at level $i$, then there could be
two edges at level $i$ adjacent to $e$ in $M^*$. In the worst case, the weight of $e$ 
could be $2^i$ and the weight of the edges in $M^*$ could be $2^{i+1}$.
The same reasoning also true for the other mapped edges mapped 
to $e$ from lower level. This was the reason we got a slightly higher approximation ratio
in the previous section. But using the new weights, we see that all the edges at a particular 
level have same weights. 

Using Lemma \ref{randomexplemma} and using linearity of expectation, we get
\begin{equation}
E_r[ \sum_{e^* \in M^*} w_r(e^*)] = \frac{\alpha-1}{\alpha \ln \alpha} \sum_{e^* \in M^*} w(e)
\label{eq3}
\end{equation}
 
We now follow the analysis used in the previous section.
Let $e$ be an edge at level $i$. Let $\phi_D^{-1}(e)$ be the edges directly mapped to $e$.
Using Claim \ref{obs2}, there can at most be two such edges. In the worst case, both
of them can be at level $i$. Summing the weight of edges directly mapped to $e$, we get

\begin{equation}
\sum_{e^* \in \phi^{-1}_{D}(e)} w_r(e^*) = 2 \alpha^{i + r} 
\label{eq4}
\end{equation}

Let $\phi_I^{-1}(e)$ be the edges indirectly mapped to $e$.
Using Claim \ref{claim2}, there can be at most two indirectly mapped edges at level 
less than $i$. Also there can at most be two such edges per level. Summing up the 
weights of indirectly mapped edges, we get
\begin{equation}
\sum_{e^* \in \phi^{-1}_{I}(e)} w_r(e^*) = 2 \sum_{j = i - 1}^{L^{min}} \alpha^{j + r}
\label{eq5}
\end{equation}

Summing up Equation \ref{eq4} and \ref{eq5}, we get the sum of the edges mapped to $e$,
$$\sum_{e^* \in \phi^{-1}(e)} w_r(e^*) = \sum_{e^* \in \phi^{-1}_{D}(e)} w_r(e^*) + \sum_{e^* \in \phi^{-1}_{I}(e)} w_r(e^*) = 
2 \sum_{j = i}^{L^{min}} \alpha^{j + r} < 2 \frac{\alpha^{i+r+1}}{\alpha-1} = \frac{2 \alpha w_r(e)}{\alpha-1}$$

Let $M_r$ be the matching obtained by our algorithm for a particular $r$.
Again following the analysis in previous section, this implies that\\
\begin{equation}
\displaystyle\sum_{e \in M_r} w_r(e) \ge \frac{\alpha-1}{2\alpha} \sum_{e^* \in M^*} w_r(e^*)
\label{eq6}
\end{equation}

\begin{theorem}
The expected approximation ratio achieved by the algorithm is $\frac{2\alpha^2 \ln \alpha}{(\alpha-1)^2}$.
The ratio is minimized when $\alpha \approx 3.512$ and the approximation ratio obtained is $\approx$ 4.9108.
\end{theorem}
\begin{proof}
Using Equation \ref{eq6} and taking expectation on both sides we get,\\
\begin{tabular}{llll}
$\displaystyle E_r[ \sum_{e \in M_r} w_r(e)]$ & $\ge$ & $\displaystyle\frac{\alpha-1}{2\alpha} E_r[\sum_{e^* \in M^*} w_r(e^*)]$ \\
  										& $\ge$ & $\displaystyle\frac{(\alpha-1)^2}{2\alpha^2 \ln \alpha} \sum_{e^* \in M^*} w(e^*)$ & Using Equation \ref{eq3}.
\end{tabular}

Using Lemma \ref{randomexplemma}, $w(e) \ge w_r(e)$ for each edge $e$ and we get\\
$\displaystyle E_r[\sum_{e \in M_r} w(e)] \ge \frac{(\alpha-1)^2}{2\alpha^2 \ln \alpha} \sum_{e^* \in M^*} w(e^*)$ \\
This implies that the expected approximation ratio obtained by our algorithm is  $\frac{2\alpha^2 \ln \alpha}{(\alpha-1)^2}$.
It achieves its minimum value when $\alpha \approx 3.512$ and the minimum expected approximation ratio
obtained is $\approx$ 4.9108.
\end{proof}

\begin{theorem}
 There exists a fully dynamic algorithm that maintains expected 4.9108-MWM  for any graph on $n$
  vertices in expected amortized $O(\log n \log \mathcal C)$ time per update.
\end{theorem}

\section{Conclusion}
We presented a fully dynamic algorithms for maintaining matching of large size
or weight in graphs. The algorithm maintains a 8-MWM with 
expected  $O(\log n \log \mathcal C)$ amortized update time.
Using a simple randomized scaling technique,
we are able to obtain  an expected 4.9108-MWM with the same update time.
The algorithm for maintaining 4.9108-MWM is the first
fully dynamic algorithm for maintaining approximate maximum weight matching.

\bibliographystyle{plain}
\bibliography{references}
\end{document}